\documentclass[runningheads]{llncs}
\input{macro.sty}

\title{Optimal Repair For Omega-regular Properties} 
\author{Vrunda Dave \inst{1}  \and Shankara Narayanan Krishna \inst{1} \and \\  Vishnu Murali \inst{2} (\Letter) \and Ashutosh Trivedi \inst{2}}
\authorrunning{V. Dave, S. Krishna, V. Murali and A. Trivedi}
\institute{{Indian Institute of Technology Bombay} \and {University of Colorado Boulder}}
\begin{document}
\maketitle

\begin{abstract}
  This paper presents an optimization based framework to automate system repair against omega-regular properties. 
  In the proposed formalization of {\it optimal repair}, the systems are represented as
  Kripke structures, the properties as $\omega$-regular languages, and the repair space as \emph{repair machines}---weighted omega-regular transducers equipped with \buchi conditions---that rewrite 
  strings and associate a cost sequence to these rewritings. 
  To translate the resulting cost-sequences to easily interpretable payoffs, we consider several aggregator functions to map cost sequences to numbers---including
  limit superior, supremum, discounted-sum, and average-sum---to define quantitative cost semantics. 
  The problem of optimal repair, then,  is to determine whether traces from a given
  system can be rewritten to satisfy an $\omega$-regular property when the allowed cost is bounded by a given threshold.
  We also consider the dual challenge of {\it impair verification} that assumes that the rewritings are resolved adversarially under some given cost restriction, and asks to decide if all traces of the system satisfy the specification irrespective of the rewritings. 
  With a negative result to the impair verification problem, we study the problem of designing a minimal mask of the Kripke structure such that the resulting traces satisfy the specifications despite the threshold-bounded impairment. 
  We dub this problem as the \emph{mask synthesis} problem.
  This paper presents automata-theoretic solutions to repair synthesis, impair verification, and mask synthesis problem for limit superior, supremum, discounted-sum, and average-sum cost semantics. 
\end{abstract}

\section{Introduction} 
\label{sec:intro}
Given a Kripke structure and an $\omega$-regular specification, the  
model checking problem is to decide whether all traces of the system satisfy the specification. 
Vardi and Wolper~\cite{vardi1986automata} initiated the automata-theoretic approach 
to model-checking by reducing the $\omega$-regular model checking problem to the 
language inclusion problem. 
If the system violates the specification, this approach returns a simple lasso-shaped 
counterexample demonstrating the violation. 
While these counterexamples often aid the designer in manually repairing the system, this 
repair process can be exhausting and error-prone. 
Moreover, different repair policies may incur different costs rendering the repair 
problem a non-trivial optimization problem.
\emph{This paper investigates a range of problems in synthesizing optimal repair policies against $\omega$-regular specification.}

As a concrete motivation for various repair problems, we consider security issues (confidentiality and availability) in manufacturing.
It is well documented~\cite{confidentiality_printer_side_channel} that acoustic
side-channels leak valuable intellectual property information during the
manufacturing process.
Consider a $3$D printer which can print either squares or
triangles.
Since the movement of the stepper motors of the printer vary based on the
design, this difference in movement leads to the printer producing different
sounds. Thus, an intruder may be able to discern the shape being printed by
observing the audio output of the system as it acts as an acoustic
side-channel.
One can model such a system as a \emph{Kripke structure}: a mockup of such systems is
represented in Figure~\ref{fig:eg_Kripke} where the label corresponds to the
state being idle ($\bot$), printing squares ($\square$), or printing triangles
($\triangle$).
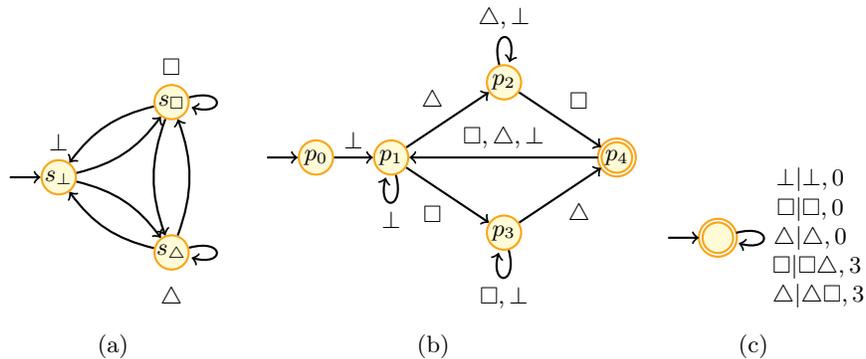
\begin{figure}[b!]
	\begin{subfigure}[t]{0.28\textwidth}
		\centering
		\begin{tikzpicture}[->,thick]
			\node[initial left,cir,initial text={}, label =
		${\bot}$]  (s0) {$s_{\bot}$} ;
			
			\node[cir, label = {$\square$}] at (1.5,1) (s1) {$s_{\square}$};
			\node[cir, label = {[yshift = -1.1cm] $\triangle$}] at
			(1.5,-1) (s2) {$s_{\triangle}$}; 
			\path (s0) edge[bend right = 20] node{} (s1);
			\path (s0) edge[bend left = 20] node{} (s2);
			\path (s1) edge[bend right = 20] node{} (s2);
			\path (s2) edge[bend right = 20] node{} (s1);
			\path (s2) edge[bend left = 20] node{} (s0);
			\path (s1) edge[bend right = 20] node{} (s0);
			\path (s2) edge[loop right] node{} (s2);
			\path (s1) edge[loop right] node{} (s1);
 		\end{tikzpicture}
		\caption{}
		\label{fig:eg_Kripke}
	\end{subfigure}
\begin{subfigure}[t]{0.40\textwidth}
	\centering
	\begin{tikzpicture}[->,thick]
		\node[initial, cir, initial text = {}] (p0) {$p_0$}; 
		\node[ cir ] at (1,0) (p1) {$p_1$};
		\node [cir ] at (2.5,1) (p2) {$p_2$};
		\node [cir] at (2.5,-1) (p3) {$p_3$};
		\node[cir,accepting] at (4,0) (p4) {$p_4$};
		\path (p0) edge node[above]{${\bot}$} (p1); 
		\path (p1) edge node[xshift = -0.2cm, yshift = 0.25cm]{$~\triangle~$} (p2);
		\path(p1) edge[loop below] node[midway, below]{$\bot$} (p1);
		\path (p2) edge[loop above] node[midway, above] {$\triangle, \bot$} (p2);
		\path (p1) edge node[yshift  = -0.25cm, xshift = -0.2cm] {$~\square~$} (p3); 
		\path (p3) edge[loop below] node[below] {$\square, \bot$} (p3);
		\path (p2) edge node[xshift = 0.25cm, yshift = 0.25cm] {$\square$} (p4);
		\path (p3) edge node[xshift = 0.25cm, yshift = -0.25cm] {$\triangle$} (p4);
		\path (p4) edge node[midway, above] {$\square, \triangle, \bot$} (p1);
	\end{tikzpicture}
	\caption{}
	\label{fig:omega_prop}
	\end{subfigure}
	\begin{subfigure}[t]{0.28\textwidth}
		\centering
		\begin{tikzpicture}[->,thick]
			\node[initial,cir,initial text={},accepting] (A) {};
			\path (A) edge[loop right] node [text
			width=1cm,align=center, right]{$\bot  | \bot, 0$ \\ $\square|\square,0$ \\ $\triangle|\triangle,0$ \\ $\square|\square\triangle,3$ \\ $\triangle|\triangle\square,3$} (A);
		\end{tikzpicture}
		\caption{}
		\label{fig:eg_transducer}
	\end{subfigure}
	\caption{ (a) Krikpe structure representing the $3$D printer System, (b) \buchi automaton $B$ specifying the property, and (c) Repair machine}
	\vspace{-2em}
\end{figure}

Suppose that the system designer wishes to protect the information that a given printer prints
only a fixed number of objects of one shape, or the sequence in which these shapes appear,  from an eavesdropper.
This specification, and a rich class of similar specifications on the observations,
can be captured using $\omega$-regular languages (see the B\"uchi automaton of
Figure~\ref{fig:omega_prop} which requires that both shapes are printed infinitely often), and one can verify if the system satisfies such a
specification using classical model checking. It is easy to see that our system
does not satisfy this property for all traces.
To repair this situation, we may wish to add spurious motor rotations to mimic
the other shape, but adding such rotations comes with a cost (say energy
or time overheads).
The choices and cost available for repair can intuitively be expressed as a
repair machine (a weighted nondeterministic transducer)  given in
Figure~\ref{fig:eg_transducer}.
For example, the label $\square | \square\triangle,3 $ represents the situation
where the repair machine modifies the observation corresponding to a square
shape by appending a spurious rotation mimicking a triangle shape with an extra
cost of $3$ units.  

A key synthesis problem, then, is to compute a minimum cost repair strategy to add these
spurious rotations such that the system after repair satisfies the
specification.
The cost of an $\omega$-sequence can be aggregated using discounted-sum, average-sum, liminf, limsup, inf, or sup.
We call this problem the \emph{repair synthesis} where the goal is given an aggregator and cost threshold, design a strategy on the nondeterministic transducer such that every trace of the system can be written to satisfy the specification with cost bounded by the given threshold. 

\begin{example}
Consider the repair machine $T$ from Figure~\ref{fig:eg_transducer} with the average-sum cost semantics and a threshold of $2$.
For every spurious motor rotation, $T$ incurs a cost of $3$ units of power. Note that a strategy of replacing every $\triangle$ with $\triangle\square$, maps $\bot \triangle^{\omega}$ to $\bot (\triangle \square)^{\omega}$ which is accepted by $B$. The mean cost of this rewrite is $3$ and is above threshold.
However, there exists a strategy that rewrites $\bot \triangle^{\omega}$ to $\bot(\triangle\triangle\triangle\square)^{\omega}$ that is accepted by $B$, with a mean cost equal to $1$. 
\end{example}

A related problem is that of \emph{impair verification} that is connected to availability vulnerabilities. Consider an attack model in the aforementioned $3$D manufacturing setting where an attacker with bounded capabilities controls the rewriting process (by introducing subtle undetectable changes in the manufacturing process) and intends to rewrite the traces in such a way that the resulting trace satisfy some undesirable behavior (to make the acoustic profile violate some regulatory norms) with a cost bounded below a threshold.
Such undesirable rewritings may impair the capabilities of the system and render it unavailable for normal use.
The impair verification problem is to verify whether the system is safe from such adversarial rewritings.

If the system is found to be vulnerable to impair and the system designer has no control over the rewriting process, a viable mitigation approach is to minimally restrict the behavior of the system to harden it against the adversarial rewriting. We formalize this problem as the \emph{mask synthesis} problem.


\vspace{1em}
\noindent \textbf{Contributions.} We consider repair machines to be specified as
weighted $\omega$-transducers. 
We study various optimal repair problems for different aggregator functions.
As we deal with reactive systems, we consider cost semantics that aggregate infinite sequence of costs to a scalar via aggregator functions discounted sum, average sum, limit superior, and supremum.
We formalize and study the following problems related to optimal repair:
\begin{itemize}
\item {\bf Repair Synthesis.} Given a system, an $\omega$-regular specification, a repair machine, and a cost semantics, decide whether there exists a strategy to rewrite traces of the system to satisfy the specification within a given threshold.

\item {\bf Impair Verification.} Given a system, an $\omega$-regular property capturing the  undesirable behaviors, a repair machine, a cost semantics, decide whether there exist a trace of the system that satisfy the undesirable behavior under adversarial rewritings within a given threshold. 

\item {\bf Mask Synthesis.} Given a system, an $\omega$-regular property (undesirable behaviors), a repair machine, a cost semantics, find a minimal restriction of the system such that no remaining trace of the system satisfy the undesirable behavior under any adversarial rewritings within the threshold. 
\end{itemize}
We characterize the complexity of repair synthesis (Theorems~\ref{thm:dsum-synth}--\ref{thm:sup-synth}) and impair verification problems (Theorems~\ref{thm:dsum-verif}--\ref{thm:sup-verif}), and for the mask synthesis problem we discuss which aggregators allow $\omega$-regular mask (Theorems~\ref{thm:dsum2}--\ref{thm:sup-mask}). 
Due to space constraints, the proofs can be found in the Appendix.

\noindent \textbf{Related Work.}
Our use of repair machines is inspired by the idea of weighted transducers studied in~\cite{concur20} for finite strings. 
The notions of robust verification and kernel synthesis studied there are templates for the impair policy verification and mask synthesis problems studied here, but the present setting requires extension of those results to the setting of $\omega$-words: this is one of the secondary contributions of this paper.
Our results imply that the results presented in~\cite{concur20} in the context of robust verification and kernel synthesis carry over to the setting of $\omega$-words for the discounted-sum and mean cost-semantics, the robust verification problem for both of these can be decided in P (cf. Theorems~\ref{thm:dsum-verif} and~\ref{thm:mean-verif}), while the robust kernel for discounted-sum cost-semantics is $\omega$-regular if the language of the Kripke structure is a cut-point language (cf. Theorem~\ref{thm:dsum2}).
Furthermore, the notion of repair synthesis, to the best of our knowledge, is yet unexplored. 

 D'Antoni, Samanta, and Singh~\cite{cav16} presented \textsc{Qlose}, a program repair approach with quantitative objectives. 
The \textsc{Qlose} approach permits rewriting syntactical expressions with arbitrary expressions while keeping the control structure of the program intact. 
In comparison, our approach permits modification of the control structure albeit with a finite set of expressions (encoded as a finite alphabet) considered for rewriting. Consequently, our setting remains decidable as opposed to repair with \textsc{Qlose} that is, in general, undecidable, and
for tractability it restricts the correctness criterion to being correct over a given set of input-output examples.
Similarly, Samanta, Olivo, and Emerson~\cite{samanta_2014_cost} considered cost-aware program repair for Turing-complete programs through the use of predicate abstraction.
However, their cost function is dependent only on the program location as opposed to more general $\omega$-traces as proposed in our work.

Jobstmann, Griesmayer, and Bloem~\cite{jobstmann_2005_program}, and  von Essen and Jobstman~\cite{von_2015_program}, studied program repair as a two-player game with qualitative $\omega$-regular objectives. Our work, in contrast,  allows quantitative notions of repair costs.
Cerny and Henzinger~\cite{emsoft11} championed for the need of partial program synthesis, which can be thought of as a repair, though its aim is to complete the given partial program, with respect to the specification. 
Although not directly related to repair, the framework of model measuring~\cite{concur13} presents a notion of distance between models; it studies the problem that given a model $M$ and specification find the maximal distance  such that all models  within that distance from $M$ satisfy the specification. Bansal, Chaudhuri, and Vardi~\cite{bansal18} study comparator automata that read two infinite sequences of weights and relate their aggregate values to compare such quantitative systems. 
Kupferman and Tamir~\cite{kupferman_2012_coping} consider the problem of cheating, where they use weighted automata and a penalty function to determine if the environment is cheating.
The penalty function considered is again a map from a pair of letters to a value and so the environment is only permitted letter-to-letter rewritings. 
In contrast, our models permits more general letter-to-string rewritings constrained with $\omega$-regular objectives.

Chatterjee et al.~\cite{chatterjee_2017_quantitative} consider the problem of solving both quantitative and qualitative objectives and define the notion of implication games where the objective is to solve both. While we provide direct proofs, Theorems~\ref{thm:dsum-synth} and~\ref{thm:mean-synth} can also be recovered from results on implication games. 

\section{Preliminaries}
\label{sec:prelims}
Let $\Sigma$ denote a finite alphabet. We write  $\Sigma^{\omega}$ and $\Sigma^*$ for the set of infinite and finite words over $\Sigma$.
We denote an empty string by $\epsilon$.

\vspace{0.5em}
\noindent\textbf{Kripke Structures.}
A \emph{Kripke structure} is a tuple $K = (S,\hookrightarrow,S_0,AP,\mathcal{L})$
where $S$ denotes a set of states, $\hookrightarrow \subseteq S \times S$ is the
transition relation, $S_0 \subseteq S$ is the set of initial states, $AP$ is
the set of atomic propositions, and $\mathcal{L}: S \rightarrow 2^{AP}$
denotes the labeling function. 
An infinite sequence of states $\pi = s_0 s_1 \ldots \in S^\omega$ is said to be a
path of the Kripke structure if $(s_i, s_{i+1}) \in \hookrightarrow$ for all
$i \in \mathbb{N}$.
Let $\Sigma = 2^{AP}$.
The labeling function applied to a path $\pi = s_0 s_1 \ldots \in
S^\omega$ defines traces $\Ll(\pi) = a_0a_1\ldots\in\Sigma^\omega$ of $K$ where for each
$i \geq 0$ we have that $a_i = \Ll(s_i)$.
We use $\Tt_K$ to indicate the set of all traces of $K$.

\vspace{0.5em}
\noindent\textbf{Omega-Regular Specifications.}
A \emph{non-deterministic \buchi automaton} (NBA) over $\Sigma$ is a tuple $A =
(Q,\Sigma, Q_0,Q_f,\delta)$, where $Q$ is a finite set of states, $Q_0 \subseteq
Q$ is the set of initial states, $Q_f \subseteq Q$ is the set of final states,
$\Sigma$ is  the finite input alphabet,  and $\delta \subseteq
Q \times \Sigma \times Q$ denotes the transition relation.   
 We define the extended transition relation $\widehat{\delta} \subseteq Q {\times} \Sigma^* {\times} Q$ in the standard fashion, i.e. $(q, \epsilon, q) \in \widehat{\delta}$ for $q \in Q$ and $ax \in \Sigma\Sigma^*$ we have $(q, ax, q') \in \widehat{\delta}$ if  there exists $q'' \in Q$ such that $ (q, a, q'') \in \delta$ and $(q'', x, q') \in \widehat{\delta}$.

A run $\rho$ over a word $w = w_0 w_1 \ldots \in \Sigma^{\omega}$ is an infinite
sequence of states $q_0, q_1 \ldots$ such that $(q_i, w_i,
q_{i+1}) \in \delta$.
A run $\rho$ is accepting iff some final state from $Q_f$ occurs infinitely
often in $\rho$.
The language defined by the automaton $A$, denoted as $L(A)$, is the set of
words $w$ over $\Sigma^{\omega}$ such that there exists an accepting run of $w$
by $A$.

\vspace{0.5em}
\noindent\textbf{Cost Aggregation Semantics.}
An aggregator function $\C: \mathbb{N}^{\omega} \to \mathbb{Q}_{\geq 0} $ maps 
infinite sequences of numbers to a scalar. Let $\tau=\tau_1 \tau_2 \dots \in \mathbb{N}^{\omega}$ with each $\tau_i \in \mathbb{N}$. 
We consider the following aggregators:
\begin{itemize}
    \item  $\mathsf{DSum_\lambda} \rmdef \overline{\tau} \mapsto \lim_{n \to \infty} \sum\nolimits_{i=1}^{n} \lambda^{i-1} \tau_i$, with discount factor  $0 \leq \lambda < 1$,
    \item $\mathsf{Mean} \rmdef \overline{\tau} \mapsto \limsup_{n \to \infty } (1/n) \cdot \sum\nolimits_{i = 1}^{n} \tau_i$,
    \item $\mathsf{Sup}
                \rmdef \overline{\tau} \mapsto \sup \{ \tau_i \mid
                i \in \Nn\}$, and 
	\item  $\mathsf{LimSup} \rmdef \overline{\tau} \mapsto \limsup \{\tau_i \mid i \in \Nn\}.$
\end{itemize}

\noindent\textbf{Quantitative Games.}
A \emph{game arena} $\Gg = (G, V_\Min, V_\Max)$ consists of a graph  $G =(V,
E, w)$ where $V$ is a finite set of vertices, $E \subseteq V \times V$ is the set of edges, $w: E \to \mathbb{N}$ is the weight function.
The sets $V_\Max$ and $V_\Min$ characterize a partition of
the vertex set $V$ such that player Min controls the edges from vertices in
$V_\Min$, while Max controls the vertices in $V_\Max$.

A play of the game $\Gg$ is an infinite sequence of vertices $\pi = \seq{v_0,
v_1, \ldots}$ such that $(v_i,v_{i+1}) \in E$ for all $i \in \mathbb{N}$.
A finite play is a finite such sequence, that is, a sequence in $V^*$.
We denote by $\last(\pi)$ the final vertex in the finite play $\pi$.
We write $\play_\Gg$ and $\fplay_\Gg$ for the set of infinite and finite plays
of the game arena $\Gg$, respectively.
A strategy of player Min in $\Gg$ is a partial function $\sigma: \fplay \to V$ defined over 
 all plays $\pi \in \fplay$ with $last(\pi) \in V_{\min}$, such that we have $(\last(\pi), \sigma(\pi)) \in E$.
A strategy $\chi$ of player Max is defined analogously.
We say that a strategy $\sigma$ is \emph{positional} if $\last(\pi)
= \last(\pi')$ implies $\sigma(\pi) = \sigma(\pi')$. 
Strategies that are not positional are called \emph{history dependent}.
Let $\Sigma_\Min$ and $\Sigma_\Max$ be the sets 
of all strategies of player \Min{} and player \Max{}, respectively.
We write $\Pi_\Min$ and $\Pi_\Max$ for the set of positional strategies of
player \Min{} and player \Max{}, respectively. 
For a game arena $\Gg$, vertex $v$ of $\Gg$ and strategy pair
$(\sigma,\chi) \in \Sigma_\Min {\times} \Sigma_\Max$, let
$\play^{\sigma, \chi}(v)$ be the infinite play starting from $v$ in which player Min and Max
play according to $\sigma$ and $\chi$, respectively.  

The weight function $w: E \to \Nat$ can be naturally extended from edges to
plays as $w: \play_\Gg \to \Nat^\omega$ as $\pi \mapsto 
c_0c_1\ldots$ where $c_i = w(v_i, v_{i+1})$ for all $i \in \mathbb{N}$.
Given an aggregator function
$\C \in \{\mathsf{DSum_{\lambda}}, \mathsf{Mean}, \mathsf{Sup}, \mathsf{LimSup}\}$,
we define the payoff of player Min to player Max for a play $\pi$ as $\C(w(\pi))$.
Depending on the choice of the aggregator function $\C \in \{\mathsf{DSum_{\lambda}}, 
	\mathsf{Mean}, \mathsf{Sup}, \mathsf{LimSup}\}$, we refer to the game 
as $\C$-game. 
In a $\C$-game, the goal of player Min is to choose her actions in such a way so
as to minimize the payoff, while the goal of player Max is to maximize the
payoff.
For every vertex $v \in V$, define the \emph{upper value} $\UVAL_\C(\Gg, v)$ as the
minimum payoff player Min can ensure irrespective of player Max's strategy.
Symmetrically, the \emph{lower value} $\LVAL_\C(\Gg, v)$ of a vertex  $v\in V$ is
the maximum payoff player Max can ensure irrespective of player Min's strategy. 
\begin{eqnarray*}
  \UVAL_\C(\Gg, v){=}\inf_{\sigma \in \Sigma_\Min} \sup_{\chi\in \Sigma_\Max} \C(w(\play^{\sigma,\chi}(v))) \\
  \LVAL_\C(\Gg, v){=}\sup_{\chi \in \Sigma_\Max} \inf_{\sigma \in \Sigma_\Min} \C(w(\play^{\sigma,\chi}(v))).
\end{eqnarray*}
The inequality $\LVAL_\C(\Gg, v) \leq \UVAL_\C(\Gg,v)$ holds for all two-player
 zero-sum games.
 A game is \emph{determined} when, for every vertex $v \in V$, the lower value and upper value are equal.
In this case, we say that the value of the game $\VAL_\C$ exists with
 $\VAL_\C(\Gg,v) = \LVAL_\C(\Gg,v) = \UVAL_\C(\Gg,v)$ for every $v\in V$.
 For strategies $\sigma \in \Sigma_\Min$ and $\chi \in \Sigma_\Max$ of players
 Min and Max,
 we define their values $\VAL^\sigma$ and $\VAL^\chi$ as
\begin{eqnarray*}
   \VAL^\sigma_\C \colon  v &\mapsto&\sup_{\chi \in \Sigma_\Max}
   \C(w(\play^{\sigma,\chi}(v))) \text{ and }\\
   \VAL^\chi_\C \colon  v &\mapsto& \inf_{\sigma \in \Sigma_\Min} \C(w(\play^{\sigma,\chi}(v))).
\end{eqnarray*}
A strategy $\sigma_*$ of player Min is called \emph{optimal} if $\VAL_\C^{\sigma_*} = \VAL_\C$.
Likewise, a strategy $\chi_*$ of player Max is optimal if $\VAL_\C^{\chi_*} = \VAL_\C$.
We say that a game is \emph{positionally determined} if both players have positional optimal strategies.

\begin{theorem}[\cite{mean_payoff_zwick,Survey_Games_Lim}] \label{thm:positional}
For $\C \in \{\mathsf{DSum_{\lambda}}, \mathsf{Mean}, \mathsf{Sup}, \mathsf{LimSup}\}$,
$\C$-games are determined in positional strategies.
The complexity of solving is in 
$\textsf{NP} \cap \textsf{co-NP}$ for $\mathsf{DSum_{\lambda}}$-games and 
$\mathsf{Mean}$-games, and, is in $P$ for $\mathsf{Sup}$-games and $\mathsf{LimSup}$-games.
\end{theorem}

The goal of the player Min in a \buchi game~\cite{chatterjee_buchi_game_alg} 
over a game arena $\Gg$ and a set $F \subseteq V$ is to choose her actions 
in  such a way that some vertex $v_f \in F$ occurs infinitely often in the play, 
while the goal of the Max player is to prevent this. 
We note from~\cite{Survey_Games_Lim} that $\mathsf{LimSup}$-games 
generalize \buchi games.
For Theorem~\ref{thm:positional} it follows that the winning region, i.e. the set 
of vertices where the player Min has a strategy to win can be computed in P.



\section{Problem Definition}
\label{sec:definition}
Just as weighted transducers extend finite state automata with outputs and costs on transitions,  NBAs can be extended to \emph{weighted non-deterministic \buchi transducers} by  adding an output word and costs to transitions.
We define a repair machine as a weighted non-deterministic B\"uchi transducer equipped with a cost aggregation.
We introduce repair machines and their computational problems.

\begin{definition}
A \emph{repair machine} (RM) $T$ is a tuple  $(Q,\Sigma, Q_0,Q_f,\Gamma, \delta, \C)$ where 
 $Q$ is a finite set of states, 
     $Q_0 \subseteq Q$ is the set of initial states, 
 $Q_f \subseteq Q$ is the set of final states,
 $\Gamma$ is the output alphabet,
  $\delta \subseteq Q \times \Sigma \times Q \times {\Gamma^{*}} \times
{\mathbb{N}}$ is the transition relation, and 
 $\C$ is the cost aggregator function.
         
\end{definition}
For a given aggregator function  $\C \in \{\mathsf{DSum_{\lambda}}, \mathsf{Mean}, \mathsf{Sup}, \mathsf{LimSup}\}$, we refer to a repair machine as $\dsum$-RM, $\mean$-RM, $\supf$-RM, $\limsupf$-RM.

	A transition $(q, a, q', w, c) \in \delta$ indicates that, the transducer on reading
the letter $a \in \Sigma$ in state $q$, transitions to state $q'$, and outputs
a word $ w \in \Gamma^*$, incurring a cost $c$ for rewriting $a$ to $w$. 
We write $q \xrightarrow{a/w}_{c} q'$ if $(q, a, q', w, c) \in \delta$.
A run $\rho$ of $T$ on $u=a_1a_2 \dots \in \Sigma^{\omega}$ is a sequence 
$\seq{q_0, (a_0, w_0, c_0), q_1, (a_1, w_1, c_1), \ldots}$
where for every $i \geq 0$ we have that $q_0 \in Q_0$ and $q_i \xrightarrow{a_i/w_i}_{c_i} q_{i+1}$.
Let $\Runs(T, u)$ be the set of runs of $T$ on $u$.
We write $\Oo(\rho)$ and $\Cc(\rho)$ for the projection on the outputs and cost sequences, i.e. $\Oo(\rho) = w_0w_1\ldots$ and $\Cc(\rho) = c_0c_1\ldots$, of a run $\rho$ of $T$.
We say that a run of $T$ is accepting if states from $Q_f$ are visited infinitely often.
We write $\dom(T)$ for the set of all words which have an
accepting run.

We define three different semantics for $T$. The function $\sem{T}(u)$ returns the set of all pairs of outputs and cost sequences over the word $u \in \Sigma^\omega$; the function $\sem{T}^\C_*(u, v)$ returns the optimal rewriting cost w.r.t the aggregator function $\C$ over $T$ for a rewriting of $u$ to $v$; and $\sem{T}^\C_\tau(u)$ returns the set of all rewritings of a word $u$ with cost bounded by a threshold $\tau \in \Real$.
\begin{eqnarray*}
\sem{T}(u) &=& \set{ 
(\Oo(\rho), \Cc(\rho)) \::\: u \in \dom(T) \text{ and  } \rho \in \Runs(T, u) },\\
 \sem{T}^\C_*(u, v)  &=& \inf \set{ \C(\Cc(\rho)) \::\: \rho \in \Runs(T, u) \text{ and } \Oo(\rho) = v}, \\
\sem{T}^\C_\tau(u)  &=& \set{ \Oo(\rho) \::\: \rho \in \Runs(T, u) \text{ and } \sem{T}^\C_*(u, \Oo(\rho)) \leq \tau) }.
\end{eqnarray*}
An example of a RM and aggregation functions is shown in Appendix~\ref{ap:eg_rm}.

\vspace{1em}
\noindent{\bf Problems of Optimal Repair}. Given the Kripke structure $K$ representing the system, the $\omega$-regular specification specified by the language $L\subseteq\Gamma^{\omega}$, a RM $T$, a cost semantics $\C \in \{\mathsf{DSum_{\lambda}}, \mathsf{Mean}, \mathsf{Sup}, \mathsf{LimSup}\}$, and a threshold $\tau \in \Qq_{\geq 0}$, the \emph{repair synthesis} problem asks if there exists a strategy of rewriting every trace $t \in \Tt_K$  to some word $w \in L$ using $T$ such that cost is at most $\tau$.

We restrict the repair policies where Player Min is restricted 
to rewrite a letter of the trace based on history and not to rely on a lookahead. 
We  give a game semantics to the repair synthesis problem as a turn-based two player game between players Min and Max that proceeds as follows. 
The game begins with player Max selecting the initial state $s_0 \in S_0$ of the Kripke structure and ends her turn.
Player Min, starts from the initial state $q_0$ of the RM and then selects a valid rewriting $w_i$ of $\Ll(s_0)$  such that $(q_0,\Ll(s_0),q'_i,w_i,c) \in \delta$ is a valid transition for some $c \in \Nat$ and changes the state of the RM to $q'_i$, she then ends her turn.
The game continues in this fashion, where player Max selects the next state $s'_i$ of the Kripke structure and Player Min selects a valid rewriting and thus the next state of the repair machine.
This turn based game proceeds indefinitely and results in Player Max selecting a trace $t \in \Tt_{K}$ and player Min selecting a word $w \in \mathbb{N}^{\omega}$.
Player Min wins the game if $w \in \sem{T}^{\C}_{\tau}(t)$, and $w \in L$, otherwise  player Max wins the game.
The existence of a winning strategy for Player Min implies the existence of a repair strategy.

    \begin{definition}[Repair Synthesis]
    Given a Kripke structure $K$ representing the system, an $\omega$-regular specification $L$, a repair machine $T$, a cost semantics $\C \in \{\mathsf{DSum_{\lambda}}, \mathsf{Mean}, \mathsf{Sup}, \mathsf{LimSup}\}$, and a threshold $\tau$ decide whether there exists a strategy to rewrite every trace $t \in \Tt_K$ to some word $w \in L$ with a cost of at most $\tau$, and if so synthesise this strategy.
    \end{definition}
      
    We also consider the dual challenge of {\it impair verification} where the system is subjected to adversarial rewritings. 
    This setting has applications in, among others, availability vulnerability detection.
    We consider an attack model where the rewritings given by the repair machine are resolved adversarially but are restricted to be within a given cost. 
    The verification problem is to decide if there exists traces of the system that satisfy an $\omega$-regular property capturing the undesirable behaviors for some such rewritings.
    The game semantics for the impair verification problem are similar to that of repair synthesis, however in the case of impair verification the player Max not only controls the selection of the next state $s'_i$, but also decides the rewriting by selecting the word $w'_i$ as well. 

  \begin{definition}[Impair Verification]
    Given a structure $K$ representing the system,  an $\omega$-regular language $L$ capturing the  undesirable behavior given as an NBA $A$, repair machine $T$, a cost semantics $\C \in \{\mathsf{DSum_{\lambda}}, \mathsf{Mean}, \mathsf{Sup}, \mathsf{LimSup}\}$, and a threshold $\tau \in \Qq_{\geq 0}$, the \emph{impair verification} problem fails if there exists a trace $t \in \Tt_K$ that can be rewritten to some word $w \in L$ with a cost of at most $\tau$ under an adversarial strategy.
    \end{definition}
  
  When one may not be able to pass the impair verification problem, it may be desirable to design a way to minimally mask the Kripke structure such that the resulting system satisfies the specifications despite the threshold-bounded impairment. In such a case, we wish to find the maximal subset $N'$ of traces which, even under adversarial rewrites,  satisfy the $\omega$-regular specification $L$.
 
\begin{definition}[Mask Synthesis]
Given a Kripke structure $K$ representing the system, an $\omega$-regular language $L$ capturing the  undesirable behavior given as an NBA $A$, repair machine $T$, a cost semantics $\C \in \{\mathsf{DSum_{\lambda}}, \mathsf{Mean}, \mathsf{Sup}, \mathsf{LimSup}\}$, and $\tau \in \Qq_{\geq 0}$, the problem of \emph{mask synthesis} is to find a maximal subset $N' \subseteq \Tt_K$ such that all traces $t \in N'$  pass the impair verification.
\end{definition}

The next three sections present our results on these three problems.

\section{Repair Synthesis} 
\label{sec:synthesis}
To solve the problem of repair synthesis, we reduce it to a related problem of \emph{threshold synthesis}. Threshold synthesis asks for a partition of the rational numbers $\Qq_{\geq 0}$ into sets $\mathbb{G}$ (good) and $\mathbb{B}$ (bad) sets such that the repair synthesis problem can be solved for all good thresholds $\tau \in \mathbb{G}$. 
Given a system $K$, the specification  $L \subseteq \Gamma^{\omega}$ represented by an NBA $B$, a repair machine $T$, and a cost semantics $\C \in \{\mathsf{DSum_{\lambda}}, \mathsf{Mean}, \mathsf{Sup}, \mathsf{LimSup}\}$,
we focus on the threshold synthesis problem: find a partition of $\Qq_{\geq 0}$  into two sets $\mathbb{G}$ and $\mathbb{B}$ such that the policy synthesis can be solved for all $\tau \in \mathbb{G}$.
We note that in the case of policy synthesis, the sets $\mathbb{G}$ and $\mathbb{B}$ are upward and downward closed respectively. 
If player Min has a winning strategy for some $\tau\in \Qq_{\geq 0}$ then she may use the same strategy for all $\tau' \geq \tau$.
Let the infimum value $\tau$ for which player Min wins be denoted as $\tau^{*}$, then $\mathbb{G} = [\tau^{*}, \infty)$ and  $\mathbb{B} = [0, \tau^{*})$. 
We call this value $\tau^*$ the optimal threshold.

\vspace{-0.7em}
\subsection{Solving the B\"uchi Games}
Our approach to compute the optimal threshold is to first restrict the choice of player Min to those where she has a strategy to win with respect to the B\"uchi objective, irrespective of the choices of Player Max on the Kripke structure.
If Player Min has no valid strategy to rewrite a trace of the system to satisfy the \buchi objective, then the optimal threshold $\tau^{*} = \infty$.
We thus consider the case when $\tau^{*} \neq \infty$ by playing a \buchi game on a game arena and then pruning it.

To construct the game arena, we first construct the synchronized product $K{\times}T{\times}B$ of $K$, $T$, and $B$. Intuitively, $K{\times}T{\times}B$ accepts those traces of the system, which have some rewriting that is in $L$.
\begin{definition}
\label{def:prod}
The synchronized product 
$K{\times}T{\times}B$ of the Kripke Structure $K = (S,\hookrightarrow, S_0, \Ll) $, the repair machine $T = (Q,\Sigma,Q_0,Q_f,\Gamma,\Delta,C)$ and the NBA $B =
(P,\Gamma,P_0,P_f,\delta)$ is a weighted (directed) graph $G^\times = (V, E, W, V_I, V_F)$, where:  
\begin{itemize}
	\item 
	$V = S \times Q \times P \times \set{1,2}$ is the set of vertices consisting of states of the system $K$, repair machine $T$, and NBA $B$, and a counter that tracks the visitation of accepting states of $T$ and $B$ (like the degeneralization construction for the generalized B\"uchi automata)
	\item 
	$E \subseteq V \times V$ is such that $((s, q, p, i), (s', q', p', i')) \in E$ if $(s, s') \in \hookrightarrow$ is a transition in $K$, for some $w \in \Gamma^{*}$ and $c \in \Nat$ transition $(q, \Ll(s), q', w, c) \in \Delta$ is in $T$, and $(p, w, p') \in \widehat{\delta}$ is a transition in $B$, and one of the following holds:
	\begin{itemize}
	\item $i = i' = 1$ and $q' \notin Q_f$
	\item $i = i' = 2 $ and $p \notin P_f$
	\item $i = 1$ and $i' = 2$ and $q' \in Q_f$
	\item $i = 2$ and $i' = 1$ and $p \in P_f$
	\end{itemize}
\item $W: E \to \Nat$ is the weight function such that
\vspace{-0.3em}
\[ 
W((q,s,p,i), (q', s', p', i')) =  \min \set{c \::\: (q, \Ll(s), q', w, c) \in \Delta};
\]
\item $V_I \subseteq V = Q_0 \times S_0 \times P_0 \times \{ 1 \} $ is the set of initial vertices;  and 
\item $V_F \subseteq V = Q \times S \times P_f \times \set{2}$ is the set of final vertices.
\end{itemize}
\end{definition}

To distinguish the choice of player Max and Min, we define a game structure $\Gg^\times$ on the product graph $G^\times$ by introducing intermediate states by appending another layer to the track counter. 
The formal construction is shown next.

The game graph  $\Gg^\times = ((\overline{V}, \overline{E},\overline{W}, V_I, V_F), \overline{V}_{\Min}, \overline{V}_{\Max})$ for product $G^\times = (V, E, W, V_I, V_F)$ is such that:
\begin{itemize}
    \item $\overline{V} = S \times Q \times P \times \set{1,2, 3}$;
    \item $\overline{E}$ is such that for $e = ((s, q, p, i)(s',q',p',i')) \in E$ we have two edges to separate the choice of the RM and the NBA from the Kripke structure:
    \begin{itemize} 
    \item 
    $e_1 = ((s, q, p, i), (s,q',p',3)) \in \overline{E}$ and 
    \item 
    $e_2 = ((s, q', p', 3)(s',q',p',i')) \in \overline{E}$;
    \end{itemize}
    with the weights $\overline{W}(e_1) = W(e)$ and  $\overline{W}(e_2) = 0$;
    \item $\overline{V}_{\Min} = S \times Q \times P \times \set{1,2}$; and 
    \item $\overline{V}_{\Max} = S \times Q \times P \times \set{3}$.
\end{itemize}
Note that the first choice is made by player Max in choosing the starting state of the Kripke structure, and in the subsequent transitions player Min reads those states and makes a choice over the rewrites. For this reason, the choice of player Max appear to be lagging by one.

We play the \buchi-game on $\Gg^\times$ with the set of accepting states as $V_F$. 
We then prune the arena to contain only those states that are in the winning region of player Min with respect to the \buchi objective, that is, the set of states where player Min has a strategy to enforce visiting \buchi states irrespective of the strategy chosen by the player Max.
We denote this pruned game arena as $\Gg$.

\subsection{Optimal Threshold for $\dsum$-RM} 
We reduce the problem of finding the optimal threshold $\tau^{*}$ for a $\dsum$-RM to the problem of finding the value of a $\dsum$-game on the game arena $\Gg$. As such we reduce the choices of selecting a trace by player Max and that of selecting a rewriting by player Min in the context of repair synthesis to choices made by the players in a $\dsum$-game over an arena $\Gg$. 
In particular, we have the following.
\begin{theorem}
\label{thm:dsum-synth}
The optimal threshold $\tau^{*}$ for the $\dsum$-RM can be computed in NP $ \cap$ co-NP via solving a $\dsum$-game on $\Gg$.
\end{theorem} 
\begin{proof}[Sketch]
We solve the $\dsum_{\gamma}$ game on $\Gg$ with $\gamma = \sqrt{\lambda}$, the value of this game corresponds to the optimal threshold $\tau^{*}$, as each edge of the synchronized product is captured by a pair of edges in $\Gg$.
For any $\varepsilon > 0$, Player Min has a strategy of following this $\dsum$ strategy, and then following the strategy of the \buchi-game such that the cost of this rewriting is $\tau^{*} + \varepsilon$.
\end{proof}

\subsection{Optimal Threshold for $\mean$-RM}
Similar to the case of the $\dsum$-RM, in the case of the $\mean$-RM, we reduce the problem of finding the optimal threshold $\tau^{*}$ to the problem of finding the value of a $\mean$-game on a game arena $\Gg$. However we note that unlike the case of the $\dsum$-RMs we also need to ensure that the mean cost cycle is co-accessible from the accepting vertices.
In particular we have the following result.
\begin{theorem}
\label{thm:mean-synth}
The optimal threshold $\tau^{*}$ for the $\mean$-RM can be computed in NP $\cap$ co-NP via solving a $\mean$ game on $\Gg$.
\end{theorem}
\begin{proof}[Sketch]
    The proof of this theorem is similar to that of Theorem~\ref{thm:dsum-synth}.
    Here, we first find a least cost mean cycle that is co-accessible by Player Min from the winning strategy of the \buchi-game on $\Gg$ (either a cycle following some $\mean$-game or the \buchi cycle itself).
    To do so we determine vertex  that is co-accessible along the $\mean$-game over $\Gg$ as well as the \buchi-game.
    Player Min then alternates between two strategies in rounds, the first, where she follows the strategy of the $\mean$-game and the second to where she follows the strategy of the \buchi-game.
    At any round $i$, she follows the strategy of the $\mean$-game until she cycles on the  co-accessible vertex $2^{i}$ many times and then follows the strategy of the \buchi-game once to return to this vertex.
    As the least cost-cycle has twice the number of edges of the synchronized product we divide the value of the $\mean$-game by two to determine the optimal threshold.
    We note that the above strategy relies on infinite memory, however Player Min can restrict the number of rounds for any $\varepsilon > 0$, and so she has a finite memory policy to guarantee repair for any threshold of $\tau^{*} + \varepsilon $.
\end{proof}

\subsection{Optimal Thresholds for $\supf$-RMs and $\limsupf$-RMs}
In the case of the $\supf$ aggregator function we first order the edges of $G^\times$ in the descending order of their weights and remove them in stages from the largest to the smallest. 
If, at any stage, the removal of edge $e$, leads to a failure of satisfying the \buchi condition, we infer that $e$ is necessary to satisfy the \buchi condition for some state in $G$.
We claim that the weight of the edge $e$ is $\tau^{*}$. 

Similar to the $\supf$ aggregator function, we start removing edges of $G^\times$ in the descending order of their weights only if they are present in an accepting cycle in the case of the$\limsupf$ aggregator function. Then, if at any stage, the removal of edge $e$, leads to a failure of satisfying the \buchi condition, we infer that the $\tau^{*}= W(e)$ and conclude that we can safely remove edges with a higher weight. 
\begin{theorem}
\label{thm:sup-synth}
    Computing optimal threshold $\tau^{*}$ for $\supf$ and $\limsupf$-RMs is in P.
\end{theorem}
\begin{proof}[Sketch]
Note that the removal of any edge $e$ from the synchronized product that causes the \buchi-objective to no longer be satisfied guarantees that all the rewrite strategies for at least one trace do not satisfy the \buchi objective.
Hence the removal prevents the satisfaction of either the acceptance of RM $T$ or the NBA $B$, and in either case, leads to a trace of the Kripke structure that cannot be rewritten to some word that is accepted by the NBA $B$. 
\end{proof}

\section{Impair Verification}
\label{sec:verification}

Given the Kripke structure $K$ representing the system, the $\omega$-regular language $L$ capturing undesirable behavior, represented as an NBA $B$, a repair machine $T$, and a cost semantics $\C \in \{\mathsf{DSum_{\lambda}}, \mathsf{Mean}, \mathsf{Sup}, \mathsf{LimSup}\}$, 
we reduce the impair verification problem to the threshold verification problem.
The threshold verification problem is to find a partition of $\Qq_{\geq 0}$  into two sets $\mathbb{G}$ and $\mathbb{B}$, such that none of the traces to system can be rewritten to a word that is in the language of $B$ for all $v \in \mathbb{G}$.
Let $\tau^{*}$ denote the infimum value for which a trace $t \in \Tt_K$ can be rewritten to some word $w \in \Gamma^{\omega}$ such that $w \in L$.
Then, the threshold verification problem is solved for any $\tau < \tau^{*}$, as $\sem{T}^\C_{\tau}(t) \not\subseteq L$ for every trace $t \in \Tt_K$.
Thus the set $\mathbb{G} = (0, \tau^{*})$ and the set $\mathbb{B} = [\tau^{*}, \infty)$ and problem reduces to finding the optimal threshold $\tau^{*}$. 

In order to find the optimal threshold $\tau^{*}$, we construct the synchronised product $G^\times = K{\times}T{\times}B$ as detailed in Definition~\ref{def:prod}. 
We prune $G^\times$ to keep only those states from where player Max has a winning strategy against the B\"uchi objective.  The construction is similar to B\"uchi games, except that the opponent has no choice.
In the following, we refer to this pruned graph as $G$.

\subsection{Optimal Threshold for $\dsum$-RM}
In the case of a $\dsum$-RM, we show that the optimal threshold $\tau^{*}$ is the minimum infinite discounted cost path in $G$. 
While it may not be possible to achieve this cost, for any $\varepsilon > 0$ we show the existence of a finite memory strategy of player Max that guarantees that some rewriting with threshold of $\tau^{*} + \varepsilon$ is in $L$.

We claim that the optimal threshold $\tau^*$ is the minimum discounted cost in $G$.
To find this value,  we associate a variable $\Vv_s$, to each vertex $v \in V$, characterizing minimum discounted 
cost among all paths starting from the state $s$. 
The minimum discounted values can then be characterized as~\cite{PutermanMDP}:
\[
	\Vv_v = \min_{(v, v') \in E} \set{W(v, v') + \lambda \cdot \Vv_{v'}}
\]
This equation can be computed by solving the following LP.
\begin{align*}
	 \text{ max } \sum_{v \in V} \Vv_v &\text{ subject to: }&
	 \Vv_v \leq W(v,v') + \lambda \Vv_{v} \text{ for all $(v, v') \in E$}.
\end{align*}

\begin{wrapfigure}{r}{0.3\textwidth}
	\begin{tikzpicture}[->,thick]
	\node[initial left,cir,initial text={}]  (p0) {$v_0$} ;
	
	\node[cir] at (1,0) (p1) {$v_1$};
	\node[cir, accepting] at (2,0) (p2) {$v_2$};
	\path (p0) edge node[midway, above]{$1$} (p1);
	\path (p1) edge node[midway, above]{$1$} (p2);
	\path (p1) edge[loop above] node[midway,above]{$0$} (p1);
	\path (p2) edge[loop above] node[midway,above]{$1$} (p2);
	\end{tikzpicture}
\end{wrapfigure}

A positional discount-optimal strategy can be computed from the solutions of these equations simply by picking a successor vertex minimizing the right side of the optimality equations. 
Observe, however, that the resulting path may not satisfy the \buchi condition. 
Consider the graph shown in the inset (right). In order to satisfy the B\"uchi objective, a run must visit the state $v_2$, while to minimize the discounted cost the strategy is to cycle in the state $v_1$ getting a discounted sum of $1$. While it is possible to achieve an $\varepsilon$-optimal discounted cost and satisfy the B\"uchi objective by looping on $v_1$ for an arbitrary number of steps before moving to the state $v_2$, no strategy satisfying the B\"uchi objective can achieve a $\dsum$ cost of $1$.


\begin{theorem}
\label{thm:dsum-verif}
The optimal threshold $\tau^{*}$ for $\dsum$-RMs can be computed in P. 
\end{theorem}


\subsection{Optimal Threshold for $\mean$-RM}

In the case of the $\mathsf{Mean}$ aggregator function, we note that only those edges that are visited infinitely often have an effect on the cost.
We say that a cycle is accepting if there exists some vertex $v \in V_F$ that occurs in the cycle.  
We let $C_1$ denote the least average cost cycle that can be reached and is reachable from some accepting cycle $C_2$. 
We use $d_1$ and $d_2$ to denote the total cost of these cycles and $n_1$ and $n_2$ to be the number of edges in each of them respectively.
We then show that $\tau^{*}$ is the mean value of cycle $C_1$.
\begin{wrapfigure}{r}{0.3\textwidth}
	\begin{tikzpicture}[->,thick]
	\node[initial left,cir,initial text={}]  (p0) {$v_0$} ;
	\node[cir, accepting] at (1,0) (p1) {$v_1$};
	\path (p0) edge[bend left] node[midway, above]{$1$} (p1);
	\path (p1) edge[bend left] node[midway, below]{$1$} (p0);
	\path (p1) edge[loop above] node[midway,above]{$1$} (p1);
	\path (p0) edge[loop above] node[midway,above]{$0$} (p0);
	\end{tikzpicture}
\end{wrapfigure}
We observe that a strategy to determine this optimal threshold requires infinite memory.
However for any $\varepsilon > 0$, there exists a finite memory strategy that is $\varepsilon$ close to $\tau^{*}$.
Consider the graph shown in the inset (right) and the following strategy adopted by Player Max.
Player Max cycles between $v_0$ and $v_1$ in rounds.
At any given round $i$, Player Max cycles on $v_0$ for $2^i$ times, and then moves and cycles once in $v_1$ and returns to $v_0$.
Observe this strategy ensures that the \buchi objective is satisfied while also ensuring the $\mean$ cost to be $0$ but requires infinite memory to keep track of the rounds.
However, Player Max can achieve a $\varepsilon$-optimal mean cost by limiting the number of rounds.
\begin{theorem}
\label{thm:mean-verif}
The  threshold $\tau^{*}$ for $\mean$-RMs can be computed in P.
\end{theorem}


\subsection{Optimal Thresholds for $\supf$-RMs and $\limsupf$-RMs}
For the $\mathsf{Sup}$ aggregator function, let $S$ be the set of values $c_i$ such that $c_i$ is the supremum of the cost of some lasso that starts from some $v_i \in V_I$ and cycles in a loop containing some $v_f \in V_F$. Let $k$ be the least element in $S$. We claim $\tau^{*} = k$.
Similar to the $\mathsf{Sup}$ aggregator function, we consider the set $S$ to contain the values $c_i$ such that $c_i$ is the supremum of the costs of the edges in the cycles that visit some $v_f \in V_F$ in the case of the $\limsupf$ aggregator function.  We then take the least of these to be the optimal threshold for the $\limsupf$-RMs. 
\begin{theorem}
\label{thm:sup-verif}
The threshold  $\tau^{*}$ for $\supf$ and $\limsupf$-RMs can be computed in P.
\end{theorem}



\section{Mask Synthesis}
\label{sec:mask}
Given a Kripke structure $K$ representing the system, an $\omega$-regular language $L$ capturing the  undesirable behavior given as an NBA $B$, repair machine $T$, a cost semantics $\C \in \{\mathsf{DSum_{\lambda}}, \mathsf{Mean}, \mathsf{Sup}, \mathsf{LimSup}\}$, and $\tau \in \Qq_{\geq 0}$, the problem of \emph{mask synthesis} is to find a maximal subset $N' \subseteq \Tt_K$ such that all traces $t \in N'$  pass the impair verification.

It is well known that every Kripke structure admits an $\omega$-regular language $N$ such that a word $u \in N$ if and only if $u \in \Tt_K$. 
Let the $\omega$-regular language of $K$ be $N$.
To solve the mask synthesis problem, we restrict the domain of the repair machine $T$ to $N$ by constructing a repair machine $T'$ using product construction and give our results on the repair machine $T'$.

\subsection{Mask Synthesis for $DSum$-RMs.} 

We show that the maximal subset $N'$ for isolated cut-point languages~\cite{QuantitativeLanguagesKLT} is $\omega$-regular.
Given a threshold $\tau \in \Qq$, the maximal subset $N'$, is the set of all words $u \in \dom(T')$, such that for every word $w' \in \sem{T}^{\mathsf{DSum}}_{\tau}(u)$ we also have $w \notin L$.
A threshold $\tau$ is $\varepsilon$-isolated for RM $T'$, if for $\varepsilon > 0$ and all accepting runs $r$ of $T'$, 
\[
\sem{T'}^{\mathsf{DSum}}_{*}(r,w) \in [0, v{-}\varepsilon] \cup [v{+}\varepsilon, \infty).
\]
It is isolated if it is isolated for some $\varepsilon$.
To prove that $N'$ is $\omega$-regular for such thresholds, we first note that isolated-cut point languages are $\omega$-regular in the context of weighted automata~\cite{ExpressivenessQL}.
We follow a similar strategy to~\cite{concur20}, and slowly unroll our synchronous product.
We note that since the repair machine is over $\omega$ strings, there must exist some $n$ such that 
\[
\mathsf{DSum}(w_0 w_1 \ldots) \leq \mathsf{DSum}(w_0 \ldots w_n) + B_n,
\]
where $B_n = V \frac{\lambda^{n}}{1-\lambda}$, where $V$ is the largest cost that is not $\infty$. Therefore if $\mathsf{DSum}(w_0, w_1, \ldots) \leq v - \varepsilon + B_n$ we can conclude that $\mathsf{DSum}(w_0, \ldots, w_n) \leq v - \varepsilon$.    

\begin{lemma}\label{lem:dsum1}
	Let $T'$ be a $\dsum$ repair machine and $\tau \in \Qq$. 
	If $\tau$ is $\varepsilon$-isolated for some $\varepsilon$, then there is $n^* \in \Nn$ such that any partial run $r$ of length at least $n^{*}$ satisfies one of the following properties:
	\begin{enumerate}
		\item $\dsum(r) {\leq} \tau {-} \varepsilon$ and $\dsum(rr') {\leq} \tau {-} \varepsilon$ for every 		infinite continuation $r'$ of $r$.
		\item $\dsum(r) {\geq} \tau {+} \frac{\varepsilon}{2}$ and $\dsum(rr') {\geq} \tau{+} \varepsilon$ for every infinite continuation $r'$ of $r$.
	\end{enumerate}
	Here, for finite $r$, $\dsum(r)$ is defined in the usual fashion except that the summation will be upto the length of $r$.
\end{lemma}

\begin{theorem}
	\label{thm:dsum2}
	Let $T'$ be a $\dsum$ repair machine, $v \in \Qq$, and $L$ an $\omega$-regular language given by
	an NBA. For all $n$, we can construct an NBA $A_n$ such that 
        $L(A_n) \subseteq L(A_{n+1})$ and
	$L(A_n) \subseteq \overline{N'} \cap \dom(T')$. 
	Moreover, if $\tau$ is $\varepsilon$-isolated, there exists $n^*$ such that $L(A_{n^*}) = \overline{N'} \cap \dom(T')$. 
\end{theorem}
For the construction of $A_n$ in Theorem~\ref{thm:dsum2}, a notion of bad and dangerous runs are defined. 
Intuitively, The bad runs are all those runs which are accepting with cost $\leq \tau$, such that the output word is not in $L$.
The dangerous runs are the finite partial runs which can be extended to bad runs.
The idea for construction of $A_n$ is to identify all the finite partial runs $r$ of length $n$ which can later be extended to bad runs.
This way we can construct a sequence of \buchi automata that better under approximate the automata for the non-robust words in the domain.
Thanks to Lemma~\ref{lem:dsum1}, we can assure that there exists a fixed point at $n^*$ such that $A_{n^*}$ recognizes all the non-robust words from $T'$. 
\subsection{Mask synthesis for $\mean$-RMs}
The mask synthesis problem for $\mean$-RMs is already undecidable for finite words~\cite[Theorem 17]{concur20} and this result carries over to the case of
$\omega$-words.
\subsection{Mask synthesis for $\supf$-RMs and $\limsupf$-RMs.}
For the $\mathsf{Sup}$-RMs, we can construct an NBA recognizing all output words with a cost greater than $\tau$ via Lemma~\ref{lem:sup1} and show that the maximal subset $N'$ is $\omega$-regular. 
The results for $\supf$-RMs can be extended carefully to only account the costs occurring in accepting loops and be used for the $\limsupf$-RMs as well.
\begin{lemma}\label{lem:sup1}
	Given a $\mathsf{Sup}$-automaton $U= (Q,q_i,\delta,\gamma)$ over $\Sigma^{\omega}$, 
	and a rational number $\tau$, we can find an NBA  $A^{> \tau} = (Q^{> \tau}, q_i^{> \tau}, \delta^{> \tau}, F^{> \tau})$ such that 
	\[
	L(A^{> \tau}) = \{ w \mid U(w) > \tau \}.
	\]
\end{lemma}

\begin{theorem}\label{thm:sup-mask}
	Let $T'$ be a $\supf$-RM, $\tau \in \Qq$ and $L$ be a $\omega$-regular language. The language of $N'$ is $\omega$-regular and we can effectively construct an NBA for it.
\end{theorem}

\section{Conclusion}
\label{sec:conclusion}
This paper presented a generalization of fundamental problems on weighted transducers and
robustness threshold synthesis for $\omega$-words.
We proposed and solved the problem of minimal cost repair formulated as
two player games on weighted transducers.
We note that this problem is similar to multi-objectives optimization where the
goal of the players is to satisfy an $\omega$-regular property while optimizing
a quantitative payoff.
We believe that the repair problem may find applications in designing mitigation
policies against side-channel vulnerability where some confidential property of
the system is leaking in the output trace, and the goal is to find a
minimum-cost repair to make the system opaque.
We also considered a related problem of impair verification that is related to availability problem where an attacker intends to rewrite the observations of the system to make it satisfy some undesirable behavior. 
Some potential future directions include the study of stochastic repair of Kripke structures and the repair of stochastic systems (Markov decision processes) with or without the knowledge of the environment.

\bibliographystyle{splncs04}
\bibliography{papers}

\clearpage
\appendix
\newcommand{\runs}{\mathsf{Runs}^{\leq n}_{T'}}
\section{Details from Section~\ref{sec:definition}}
\subsection{Example Repair Machine}
\label{ap:eg_rm}
	Consider the repair machine $T$ shown in Figure \ref{eg1} and an input $u=ba(ab)^{\omega}$. 
	The accepting run for $u$ visits states   
	$\ell_0, \ell_1, (\ell_2,\ell_3)^{\omega}$.  
	So we have  $\sem{T}(u)=(cc(ccd)^{\omega},20(14)^{\omega})$. 
	If $T$ is a $\dsum$-RM, with the value of $\lambda=\frac{1}{2}$, we have 
		$\sem{T}^{\mathsf{Dsum}}_{*}(ba(ab)^{\omega}, cc(ccd)^{\omega})=2+\frac{1}{4}+\frac{1}{2}+\frac{1}{16}+\frac{1}{8}+ \dots=3.
$	Likewise, 
	considering $\C =\mathsf{Sup}$, we note that $\sem{T}^{\mathsf{Sup}}_{*}(ba(ab)^{\omega}, cc(ccd)^{\omega})=4$, and for
	$C=\mathsf{LimSup}$, we have  $\sem{T}^{\mathsf{LimSup}}_{*}(ba(ab)^{\omega}, cc(ccd)^{\omega})=4$. 
	Lastly, for the case of $C=\mathsf{Mean}$, we have 
	$\sem{T}^{\mathsf{Mean}}_{*}(ba(ab)^{\omega}, cc(ccd)^{\omega})=\limsup\{2, 1,1,\frac{7}{4},\frac{8}{5}, 
	\frac{12}{6},\frac{13}{7},\frac{17}{8} \dots\}=2.5.
$
	
	\begin{figure} [h]
	\begin{center}
		\begin{tikzpicture}[->,thick]
		\node[initial above,cir,initial text={},accepting] at (-6,-1) (A)
		{$\ell_0$} ; 
		\node[cir,initial text={}] at (-3,-1) (B)
		{$\ell_1$} ; 
		
		\node[cir,accepting] at (0,-1) (C) {$\ell_2$};
		\node[cir] at (3,-1) (D) {$\ell_3$};

		\path (A) edge[loop below] node [midway,below] {$a|d,1$}(A);
		\path (A) edge node [midway,below] {$b|c$,2}(B);
		\path (B) edge node [midway,below] {$a|c$,0}(C);
		\path (C) edge node [midway,below] {$a|c$,1,$~~b|c$,1}(D);
		
		\path (D) edge[bend right=50] node [midway,below] {$b|cd$,4}(C);      
		\path (D) edge[bend right=50] node [midway,above] {$a|\epsilon$,1}(B);      
		
		\end{tikzpicture}                                
		\caption{An example of a repair machine.}
		\label{eg1}
	\end{center}
\end{figure}
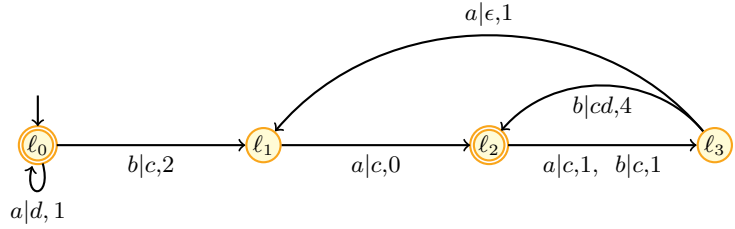

\section{Proofs from Section~\ref{sec:synthesis} }
\subsection{Proof of Theorem~\ref{thm:dsum-synth}}
\begin{proof}
    Consider the game arena $\Gg$, and the $\dsum_{\gamma}$-game on $\Gg$ with $\gamma = \sqrt{\lambda}$.
 Let the vertex $v_0 \in V_I$ be the initial vertex with the maximum value when the game is determined with a value $\tau_{r}$. 
 That is, $\VAL_{\dsum}(v_0) \geq \VAL_{\dsum}(v_i)$ for all vertices $v_i \in V_I$, and $\VAL_{\dsum}(v_0) = \tau_{r}$.
 We claim that the value of this game is $\tau^{*}$.
The proof of this is in two parts.
\begin{itemize}
    \item [(1)]
    Let the optimal play by the players from $v_0$ correspond to a sequence of vertices and edges denoted as $\seq{v_0, e_0, v'_0, e'_0, v_1, e_1, v'_1, e'_1  \ldots}$. We have the value of  $\tau_{r} = \overline{W}(e_0) + \gamma \overline{W}(e'_0) + \gamma^2 \overline{W}(e_1) + \ldots $ as the discounted sum of this path.
    The selection of edges of the form $e'_i$ are by player Max and have a cost $0$ while the selection of edges of the form $e_i$ are by player Min and have a noon-zero cost $\overline{W}(e_i)$. 
    As this is the optimal play, there is no strategy by which player Min can unilaterally decide on a better strategy.
    Thus, the value $\tau_{r} = \overline{W}(e_0) + \gamma^2 \overline{W}(e_1) + \ldots $, with $\gamma = \sqrt{\lambda}$.
    The selection of edges $e'_i$ by player Max corresponds to the selection of a trace $t \in \Tt_K$. 
    A valid rewriting of this trace has a $\dsum$ cost of at most $\tau$, where the value $\tau = \overline{W}(e_0) + \lambda \overline{W}(e_1) + \ldots $, as this corresponds to the optimal selection of a rewriting of $t$. 
Thus the least cost rewrite of $t$ has a value of at least $\tau_{r}$ contradicting our assumption.
\item [(2)]
For any $\varepsilon > 0$, the value $(\tau_{r} + \varepsilon) \in \mathbb{G}$.
Let $W$ be an upper bound on the weights in $\Gg$, then for any $\varepsilon > 0$, there exists some $k\in \mathbb{N}$ such that $\lambda^{k} \frac{W}{1 - \lambda} \leq \varepsilon$. 
We thus follow the optimal discounted strategy for $k$ steps and then choose the strategy of the \buchi-game to follow a path to an accepting cycle and cycle thereafter.
Thus the resulting path has a discounted cost at most $\tau_{r} + \varepsilon$. 
As this path also wins the \buchi-game on $\Gg$ with discounted cost at most $\tau_{r} + \varepsilon$, the corresponding output word $w$ must also be in $L$.
\end{itemize}
We note that the time taken to construct graph $\Gg$ is polynomial in the size of $K,T$ and $B$ as constructing the synchornized product, solving the \buchi game and pruning all take polynomial time.
However, from Theorem~\ref{thm:positional}, we note that solving the $\dsum_{\gamma}$-game on  $\Gg$ is in NP $\cap$ co-NP. 	 
\qed
\end{proof}

\subsection{Proof of Theorem~\ref{thm:mean-synth}}

\begin{proof}
 We play a \buchi-game on $\Gg$ and find the set of \buchi cycles for each vertex $v_f \in V_F$. We note that the \buchi-game is determined from Theorem~\ref{thm:positional} and so the optimal play by both players decides on this cycle. We then find the accepting vertex $v_j \in V_F$, with the least mean value of the \buchi cycle that can be reached from $v_f$ and denote its mean weight as $M_1$.
 
 We then play the $mean$-game on $\Gg$ and determine the positional strategy for each vertex $v \in V$.
 Finally, we check if a vertex in the \buchi cycle is co-accessible according to the $\mean$-game strategy. Let the \buchi cycle for a vertex $v_f \in V_F$ be denoted as $C_f$. We find if each vertex $v_i$ in $C_f$ is co-accessible based on the $\mean$-game. If so, we denote the cost of that cycle as $M_2$ and note that its mean cost if it is less than $M_1$, otherwise we use $M_2$ to denote the cost $M_1$. 

We show that value $\tau^{*}$ is twice $M_R$, where $M_R$ is the maximum value $M_1$ for every accepting state $v_f \in V_F$.

If the accepting state does not satisfy the second condition, we note that player Max has a strategy to ensure that the $\mean$-game strategy cannot be taken for any vertex in the \buchi cycle of vertex $v_f$ such that it returns to the same vertex. 
Then we conclude that there exists a trace $t$ which forces the synchronized product into this accepting state, and prevents any of the vertices from following a strategy to minimize the mean cost. 
If some vertex $v$ in the cycle is co-accessible, we instead note that player Min has a strategy to achieve a minimum value of at least $M_2$ regardless of the choices of player Max and similarly has a way to ensure that the \buchi cycle is taken regardless of the choices of player Max.
Thus player Min follows a strategy of alternating between the two cycles in rounds.
At any round $i$, player Min follows the $\mean$-game strategy for $2^i$ repetitions and then follows the \buchi cycle once. Such a strategy leads to a mean cost of $M_2$ while still satisfying the \buchi acceptance condition. 
As the game arena $\Gg$ has two edges for each edge in $G^\times$, the value of $\tau^{*} = 2M_R$.
For the complexity, 
We note that the time taken to construct graph $\Gg$ is polynomial in the size of $K,T$ and $B$ as constructing the synchornized product, solving the \buchi game and pruning all take polynomial time.
Further, reachability and \buchi-games are in $P$ from \cite{Jones_reach}, and Theorem~\ref{thm:positional}, while $\mean$-games are in NP $\cap$ co-NP. Checking if each vertex repeats in its corresponding strategy enumerates at most all the vertices. Checking the mean costs of the \buchi cycles and the least cost cycles is polynomial in the number of edges. 
Thus the problem is in NP $\cap$ co-NP.
\qed
\end{proof}
\subsection{Proof of Theorem~~\ref{thm:sup-synth}}
\label{ap:sup-synth}
\begin{proof}
We first show that the value $\tau_{r}$ is equal to the value $\tau^{*}$. 
First $\tau^{*}\geq \tau_{r}$. 
Let us suppose otherwise, then there must be a strategy to rewrite every trace $t \in \Tt_K$ to some word $w \in L$ with a cost that is less than $\tau_{r}$. 
We note from our procedure that the removal of the edge $e$ leads to a failure in satisfying the \buchi conditon, and so there must be some trace $t \in \Tt_K$ that cannot be rewritten to some word $w \in L$ contradicting our claim.
Now we show that $\tau^{*}\leq \tau_{r}$. 
If $\tau^{*}> \tau_{r}$, then there exists some trace $t \in \Tt_K$, such that 
\[ 
\tau_{r} < \sem{T}^{\mathsf{Sup}}_{*}(t,w) \leq \tau^{*}
\]
for all $w \in \sem{T}^{\supf}_{*}(t)$ and $w \in L$.   
As we removed the edges in descending order of weights, and are still able to satisfy the \buchi condition, for every trace $t \in \Tt_K$, there must be some word $w'$ such that $w' \in L$ and $\sem{T}^{\mathsf{Sup}}_{*}(t,w') \leq \tau_{r}$, which is again a contradiction.

We note that the size of $G^\times$ is polynomial in the sizes of $K,T$ and $B$. As such the number of edges are also polynomial the sizes of $K,T$ and $B$. We follow the procedure from above and note that we can at most remove all the edges. Lastly, from Theorem~\ref{thm:positional}, we have the complexity of solving \buchi games to be in P. Thus the problem is in P.
The proof follows in a similar fashion for the $\limsupf$ aggregator function.
\qed
\end{proof}


\section{Proofs from Section~\ref{sec:verification}} 

\subsection{Proof of Theorem~\ref{thm:dsum-verif}}

\begin{proof}
First we show that the optimal threshold $\tau^{*}$ is the minimum infinite discounted cost path.
Let value of the minimum discounted cost path be $\tau_{r}$, and let us suppose $\tau^{*} < \tau_{r}$, then there must exist a trace $t \in \Tt_K$ and a word $w \in L$, such that $\sem{T}^{\mathsf{DSum}}_{*}(t,w) < \tau_{r}$. 
As the trace $t$ has a rewriting that is in $L$, it must be accepted by the synchronized product, however if $t$ is accepted by $g^\times$, then it must follow an infinite path in $G^\times$.
The discounted cost of this infinite path must be at least as large as $\tau_{r}$ which is a contradiction.

 Let $W$ be an upper bound on the weights in the graph $G^\times$. Then for any $\varepsilon > 0$, there exists some $k\in \mathbb{N}$ such that $\lambda^{k} \frac{W}{1 - \lambda} \leq \varepsilon$.  
 We follow the discounted-optimal strategy obtained from the above optimality equation for $k$ steps and then choose an arbitrary strategy to an accepting cycle and follow an accepting cycle thereafter. 
 It follows that the resulting path is an accepting path with cost bounded by $\tau^{*} + \varepsilon$.
  It is easy to see this : consider a path with all the edges having cost $W$. 
 The discounted sum for this path is $\frac{W}{1-\lambda}$. 
 If we follow the discounted strategy for the first $k$ steps on the path as defined above, then any of its infinite extensions will cost no more than $\tau^{*} + \lambda^k \frac{W}{1-\lambda} \leq \tau^{*} + \varepsilon$. 
 We note that this is an accepting path in $G^\times$, and so for any $\varepsilon > 0$, $\tau_{r} + \varepsilon \in \mathbb{B}$ .
 
 We note that the size of the graph $G^\times$ is polynomial in $K,T$ and $A$.
 Solving the \buchi game on this graph and the above linear program is in $P$.
 \qed
\end{proof}

	
\subsection{Proof of Theorem~\ref{thm:mean-verif}}
\begin{proof}
    First we claim that the optimal threshold $\tau^{*}$ is the value of the average cost of cycle $C_1$, where $C_1$ is the least cost cycle that can be reached and is reachable from some accepting cycle $C_2$.
	Suppose $\tau^{*}$ is lesser than the mean cost of $C_1$. 
	Then there must exist some edges that are visited infinitely often such that their average cost is less than the average cost of $C_1$.
	The graph $G^\times$ has a finite number of edges, so these edges must be a part of some cycle as they are visited infinitely often. 
	However this cycle must also be reachable and reach some accepting cycle.
	This must imply that there exists a cycle with an average cost that is strictly less than the average cost of $C_1$ which contradicts our assumption.

    If $C_1$ and $C_2$ are the same, \textit{i.e.}, $C_1$ is an accepting cycle then the value of $\tau^{*}$ is equal to the average cost of $C_1$. Without loss of generality we can assume that the two cycles $C_1$ and $C_2$ are incident at a common point (by collapsing the path in between them). 
    We can collapse cycles $C_1$ and $C_2$ into a common point by extending the cycle $C_2$ to constitute a larger cycle that passes through $C_1$ and returns to $C_2$. 
	Moreover, it does not affect the least mean cost of the graph.
	We now prove $\tau^{*} = \frac{d_1}{n_1}$ by giving a strategy to find an accepting path of the graph with mean cost as $\frac{d_1}{n_1}$.
	To find the value of $\tau^{*}$ when $C_1$ and $C_2$ are not the same, we consider the strategy of cycling between $C_1$ and $C_2$ in rounds. At any round $i$, we cycle $2^{i}$ many times in $C_1$ and once in cycle $C_2$. We can write the value of the average computed after round $i$ as 
	
\[a_i = \dfrac{(1 + 2 + \ldots + 2^{i}) d_1 + i \cdot d_2 }{(1 + 2 + \ldots + 2^{i}) n_1 + i\cdot n_2} = \dfrac{(2^{i+1} - 1) d_1 + i\cdot d_2 }{(2^{i+1} - 1) n_1 + i\cdot n_2} \]

From the properties of limits of real sequences we note that the sequence $a_0, a_1, \ldots $ converges to $\frac{d_1}{n_1}$, and thus we have $\tau^{*} = \frac{d_1}{n_1}$.
We note that the sequence $a_1, a_2 \ldots$ is a Cauchy sequence and so for any $\varepsilon > 0$, there exists some $M$ such that for all $i,j \geq M$, we have $|a_i - a_j| \leq \varepsilon$. 
Thus for any $\varepsilon > 0$, there exists some $k$ such that $a_k = \tau^{*}+ \varepsilon$. 
We can thus follow a strategy where we alternate cycling $2^{k+1}$ times in $C_1$ and once in $C_2$ to approximate $\tau^{*}$ to within $\varepsilon$.

Now notice that the graph $G^\times$ is polynomial in the sizes of $K,T$ and $A$.
We find the least mean cost cycle that is co-accessible from an accepting vertex as follows.
Secondly we use Karp's algorithm~\cite{karp_mean_value} to find the cost of the least mean cycle and note that this cycle can be reached and is reachable from an accepting vertex. 
The mean value of this cycle corresponds to the optimal threshold and so this can be done in P.
\qed
\end{proof}
\subsection{Proof of Theorem~\ref{thm:sup-verif}}
    Computing optimal threshold $\tau^{*}$ for $\supf$ and $\limsupf$-RMs is in P.
\label{ap:sup-verif}
	\begin{proof}
 We have $\tau^{*} = k$ as follows, first $\tau^{*} \geq k$. Let us suppose otherwise, then there must exist some path that starts from some $v_i \in V_I$ and cycles infinitely often on state $v_f$ such that the supremum of that path is less than $k$, but this would contradict our assumption. We know that such a $k$ exists as it corresponds to the supremum of some lasso that starts from some $v_i \in V_I$ and visits infinitely often some $v_f \in V_F$.
 
 Now we show the value $k$ can be computed in $P$ as follows. The graph $G^\times$ is polynomial in $K,T$ and $A$ and has a polynomial number of edges.
 Pruning this graph to contain only those vertices which can reach the accepting state can also be done in polynomial time.
 While we seek to find the supremum value in all lassos, we do not need to enumerate all the lassos. We simply need to check the supremum value in the problem of reachability to an accepting state, and then check reachability from that accepting state to itself from all vertices $v_i \in V_I$ and store the maximum value for each.
Thus the two step reachability problem can be solved in $P$.
The proof follows in a similar fashion for the $\limsupf$-RMs.
\qed
\end{proof}

\section{Proofs from Section~\ref{sec:mask}} 

	\subsection{Proof of Lemma~\ref{lem:dsum1}}
	\label{app:lem-dsum1}
\begin{proof}
	Let $r$ be a partial run of length $n$ of $T$. Since $T$ is trim, there exists a continuation $r'$ of $r$, and moreover we have $\dsum(rr') = \dsum(r) + \lambda^n \dsum(r')$.
	We have $\dsum(r') \leq \Sigma_{i=0}^{\infty} \lambda^i W = \frac{W}{1 - \lambda}$ where $W$ is the largest weight of any transition in $T'$. Let $B_n = \lambda \frac{W}{1 - \lambda}$.
	Let $n^*$ be the smallest non-negative integer such that $B_{n^*} \leq \varepsilon /2$. (it exists since $b_n$ is strictly decreasing of limit 0).
	Assume that the length of $r$ is greater than $n^*$ i.e. $n \geq n^*$.
	As a consequence $B_n \leq B_{n^*}$.
	Since $v$ is $\varepsilon$-isolated, we have two cases:
	\begin{enumerate}
		\item If $\dsum(rr') \leq \tau - \varepsilon$ then $\dsum(r) \leq \tau - \varepsilon$ since $\dsum(r) \leq \dsum(rr')$ by non-negativity of the weights of $T'$.
		\item If $\dsum(rr') \geq \tau + \varepsilon$ then $\dsum(r) \geq \tau + \varepsilon - \lambda^n \dsum(r')$. Moreover, $\lambda^n \dsum(r') \leq B_n \leq B_{n^*} \leq \varepsilon/2$ by construction. So $-\lambda^n \dsum(r') \geq -\varepsilon/2$ which implies $\dsum(r) \geq \tau + \varepsilon/2$. 
	\end{enumerate}
 
 	We have shown that either of the above cases occurs since $\tau$ is $\varepsilon$ isolated. 
 	Next we prove that, for all continuation $r'$ of $r$ we have (i) $\dsum(r) \leq \tau - \varepsilon$ implies $\dsum(rr') \leq \tau-\varepsilon$ and (ii) $\dsum(r) \geq \tau + \varepsilon /2$ implies $\dsum(rr') \geq \tau + \varepsilon$.
 	In the first case, assume by contradiction that some continuation $r'$ of $r$ satisfies $\dsum(rr') \geq v + \varepsilon$. As a consequence  $\lambda^n \dsum(r') \geq 2 \varepsilon$, which is impossible since $\lambda^n \dsum(r') \leq B_n \leq B_{n^*} \leq \varepsilon /2$.
 	In the second case, if $\dsum(r) \geq \tau + \varepsilon/2$ then any continuation $r'$ of $r$ satisfies $\dsum(rr') \geq \dsum(r) > \tau + \varepsilon /2$. 
 	Since $\tau$ is $\varepsilon$-isolated, we get $\dsum(rr') \geq v+\varepsilon$.
 	\qed
	\end{proof}

\subsection{Proof of Theorem~\ref{thm:dsum2}}\label{app:thm-dsum2}

	Let $T'$ be a weighted omega transducer with $\dsum_\lambda$ aggregator function, $\tau \in \Qq$, and $L$ an $\omega$-regular language $L$ given by an NBA. For all $n$, we can construct an NBA $A_n$ such that:
\begin{enumerate}
	\item $L(A_n) \subseteq L(A_{n+1})$
	\item $L(A_n) \subseteq \overline{N'} \cap \dom(T')$
\end{enumerate}
Moreover, if $\tau$ is $\varepsilon$-isolated, there exists $n^*$ such that $L(A_{n^*}) = \overline{N'} \cap \dom(T)$. 

\begin{proof}
	Let $B_n = \lambda^n \frac{W}{1-\lambda}$ as defined above where $W$ is the maximum weight of the edges. 
	A run $r$ on a pair $(u_1, u_2)$ is called bad if $\dsum(r) \leq \tau$ and $u_2 \in L$ and $r$ is accepting.
	Note that $u_1 \notin N'$.
	A finite run $r$ is called dangerous if $|r| \geq n$ and $\dsum(r)  \leq \tau - B_n$.
	A dangerous run $r$ can possibly be extended to a bad run $rr'$.
	For the extended bad run, we only need to check that the output satisfies the $\omega$-regular property represented by $L$, since the cost cannot grow more than $\tau$. 
	We exploit this idea for construction of automaton $A_n$.
	Intuitively, $A_n$ recognizes all those words for which there exists a dangerous prefix run of length $n$ at most which can be extended to a bad run (Figure~\ref{fig:dsum-kernel}).
	
	\begin{figure}[t]
		\centering
		\includegraphics[scale=0.4]{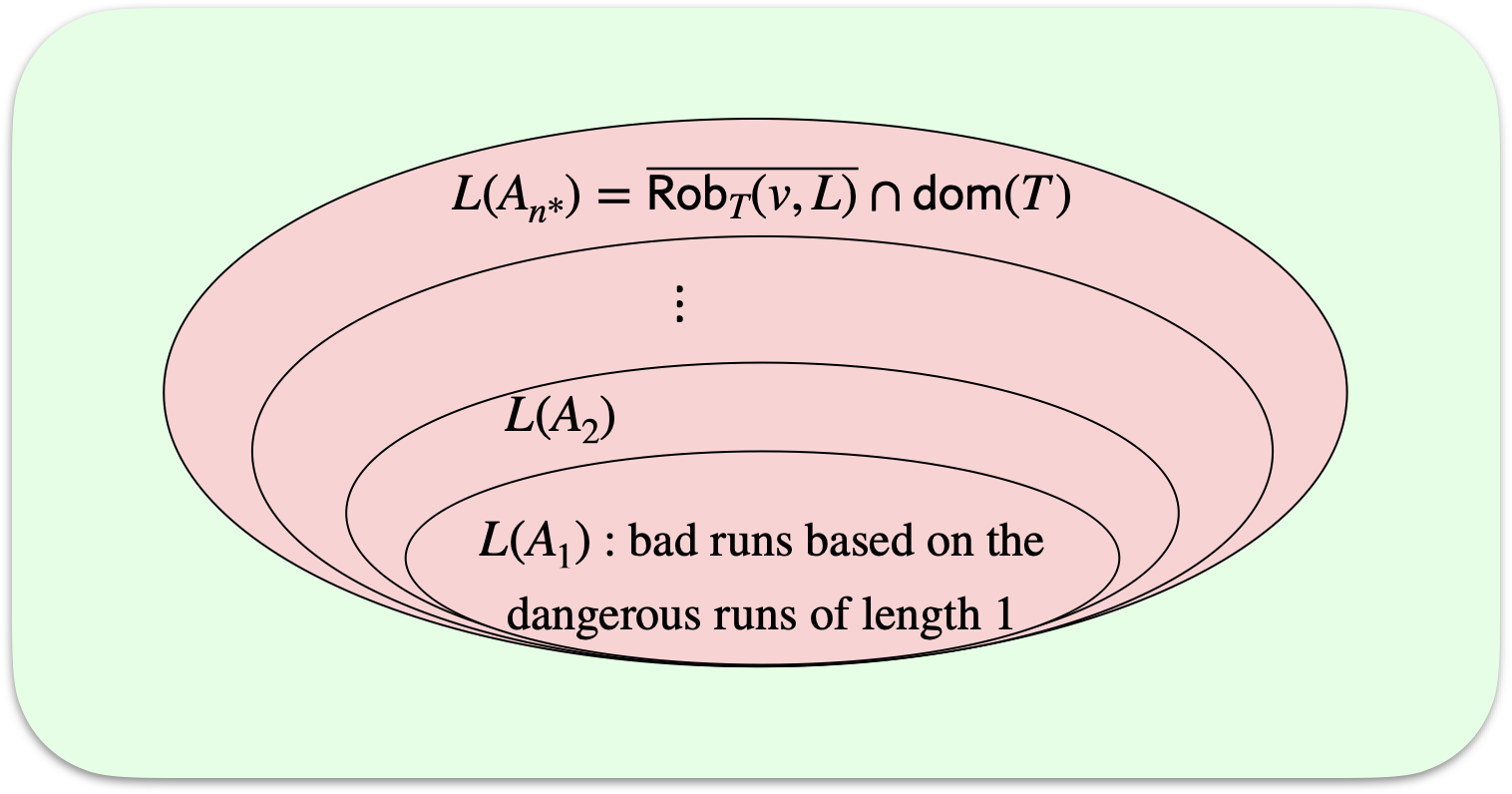}
		\caption{Successive approximations of the maximal subset.}
		\label{fig:dsum-kernel}
	\end{figure}
	
	Let $\runs$ be the runs of $T'$ of length at most $n$, and $Q$ its set of states. 
	We assume that for all  $(u_1, u_2)$ recognized by $T'$, $u_2 \in L$ holds. 
	This can be assured by taking synchronized product of $T'$ with the automaton $B$ recognizing $L$. 
	Let us formally construct $A_n$:
	Its set of states is union of states of $T'$ and $\runs$. 
	Transitions are defined as follows: for all runs $r \in \runs$ of length $n-1$ at most ending in some state $q$, for all $\sigma \in \Sigma_\varepsilon$, if there exists a transition $\delta$ of $T'$ from state $q$ on reading $\sigma$ then we have the transition $r \xrightarrow{\sigma} r\delta$ in $A_n$. 
	From any run $r \in \runs$ of length $n$: (i) if $r$ is not dangerous, we do not have any outgoing transitions in $A_n$, (ii) if $r$ is dangerous, we have a transition to its end state from $Q$ while reading $\varepsilon$. 
	Furthermore, we have all the transitions $\delta$ of $T'$ be present in the transitions of $A_n$. 
	Accepting states are the same as that of $T'$.
	
	It is easy to see that $L(A_n) \subseteq L(A_{n+1})$ for all $n$.
	Let $w \in L(A_n)$ with an accepting run $\rho = \rho_1 \rho_2$ such that $\rho_1 \in (\runs)^*$ and $\rho_2 \in Q^\omega$. 
	By construction of $A_n$, the last state of $\rho_1$, say $r_n$,  represents the dangerous run of length $n$ and from its last state, say $q$, we have $\varepsilon$ transition to $\rho_2$. 
	Moreover, $r_n \rho_2$ is a bad run since $r_n$ is a dangerous run and $\rho_2$ satisfies the regular property.
	Since $r_n$ was dangerous at step $n$, $\dsum(r_n) \leq v - B_n$, and by definition of $B_n$, $B_n > B_{n+1}$. Thus, 
	\[
	\dsum(r_nq) \leq \dsum(r_n) \leq \tau - B_n \leq \tau - B_{n+1}.
	\]
	Hence $r_nq$ is dangerous at step $n+1$, and we get an accepting run of $w$ in $A_{n+1}$.
	
	Suppose $\tau$ is $\varepsilon$-isolated for some $\varepsilon > 0$. 
	Let $n^*$ be the smallest non-negative integer such that $B_{n^*} \leq \varepsilon /2$ (same as in above proof). 
	We wish to show $L(A_{n^*}) = \overline{N'} \cap \dom(T')$. One side direction is easy that $L(A_{n^*}) \subseteq \overline{N'} \cap \dom(T')$. 
	For the other side proof, let $w \in \dom(T')$ and $w \notin \overline{N'}$. 
	This implies that there exists $(w, w',c) \in \sem{T'}$ such that its accepting run $r$ has $\dsum(r) \leq \tau$ and $w' \notin L$. 
	We can say that $r$ is bad. 
	Since $\tau$ is $\varepsilon$-isolated, $\dsum(r) \leq \tau -\varepsilon$. 
	
	By Lemma~\ref{app:lem-dsum1}, we have $\dsum(r[:n^*]) \leq \tau - \varepsilon$ where $r[:n^*]$ represents the prefix of $r$ of length $n^*$.
	We know that $B_{n^*} \leq \varepsilon/2$, therefore 
	\[
	\dsum(r[:n^*]) \leq \tau -\varepsilon \leq \tau - \varepsilon/2 \leq \tau - B_{n^*}.
	\]
	Hence $r[:n^*]$ is a dangerous partial run, its continuation $r$ will have a corresponding accepting path in $A_{n^*}$ by the definition of $A_{n^*}$.
	\qed
\end{proof}

\subsection{Proof of Lemma~\ref{lem:sup1}}\label{app:lem-sup1}
\begin{proof}
		We construct $A^{> \tau} = (Q^{> \tau}, q_i^{> \tau}, \delta^{> \tau}, F^{> \tau})$ as follows :
	\begin{itemize}
		\item  $Q^{> \tau} = Q \times \{0,1\} $;
		\item $q_i^{> \tau} = (q_i,0) $;
		\item  $((q,i), \sigma , (p,j) ) \in \delta^{> \tau} $ if $(q, \sigma, p) \in \delta$ and one of the following is satisfied:
		\begin{itemize}
			\item $i = 0$, $j = 0$ and $ \gamma(q,p) \leq \tau$,
			\item $i = 0$, $j = 1$ and $\gamma (q,p) > \tau$,
			\item $i = 1$ and $j = 1$; and 
		\end{itemize}
		\item $F^{> v} = Q^{> v} \times {1}$.
	\end{itemize}
	It is clear that $A^{> \tau}$ accepts a word $w \in \Sigma^{\omega}$ iff $ U(w) > \tau$.
	\qed
\end{proof}

\subsection{Proof of Theorem~\ref{thm:sup-mask}}\label{app:sup-mask}
\begin{proof}
	We note that $N'$ contains the set of words $u$ such that for every word $w \in \Gamma^{\omega}$, if $\sem{T'}^{\mathsf{Sup}}_{*}(w_1, w_2) \leq \tau$, we have $w \notin L$.
	The complement of $N'$ can be defined as  
	$$\overline{N'} = \{u \mid \exists w, \sem{T'}^{\mathsf{Sup}}_{*}(u, w) \leq \tau \wedge w \in L\}.$$
	We show it is $\omega$-regular.
	Let $A$ be the NBA recognizing $L$.
	First, we take the synchronizing product of $T'$ and $A$, denoted by $T' \otimes a$, where $A$ works on outputs of $T'$ and the product accepts those words whose rewriting is in $L$.
	Hence,
	$$\overline{N'} = \{w_1 \mid \exists w_2, \sem{T' \otimes A}^{\mathsf{Sup}}_{*}(w_1, w_2) \leq \tau\}.$$
	We can project out the outputs from product transducer and get a $\sup$ automaton $U$.
	Now, the problem reduces to 
	$\overline{N'} = \{w_1 \mid U(w_1) \leq \tau\}$.
	Complementing again, we get $N' = \{w_1 \mid U(w_1) > \tau\}$.
	Now, we can apply Lemma~\ref{lem:sup1} to get an NBA for the language of robust words.
	\qed
\end{proof}


\end{document}